\documentclass[12pt,a4paper]{amsart}%
\usepackage{amsmath}
\usepackage{bbm}
\usepackage{amsfonts}
\usepackage{amssymb}
\usepackage{amsthm}
\usepackage{verbatim}
\usepackage{enumitem}

\newcounter{all}

\newtheorem{theorem}[all]{Theorem}
\newtheorem{lemma}[all]{Lemma}

\def\codomainbool{\mathcal B}
\def\fclone#1{\langle#1\rangle}  
 
\def\fclonelim#1{\langle#1\rangle_\omega}

\def\LSM{\mathsf{LSM}}
\def\R{\mathbb{R}}
\def\Rnonneg{\R^{\geq0}}
\def\calF{\mathcal F}
\def\veca{\boldsymbol a}
\def\vecb{\boldsymbol b}

\def\vecw{\boldsymbol w}
\def\vecx{\boldsymbol x}
\def\vecy{\boldsymbol y}
\def\vecz{\boldsymbol z}
\def\vecone{\boldsymbol 1}

\def\EQ{\mathrm{EQ}}
\def\IMP{\mathrm{IMP}}

\def\tr{\mathrm{tr}}

\def\omegadef{$\mathrm{pps}_\omega$-definable}
\def\ppdef{pps-definable}  
\def\ppdefinable{\ppdef}
\def\ppformula{pps-formula}
\def\ppformulas{pps-formulas}

\def\omegadefinability{$\mathrm{pps}_\omega$-definability}
\def\dom{\{0,1\}}
\def\Fhat{\widehat F}

\def\myset{\mathcal{C}}

\title{LSM is not generated by binary functions}
\author{Colin McQuillan}

\address{Colin McQuillan\\
Department of Computer Science\\
University of Liverpool\\Ashton Building\\
Liverpool L69 3BX\\
United Kingdom.}

\thanks{This work was supported by an EPSRC doctoral training grant.}

\begin{document}

\maketitle

\noindent {\bf The material in this note is superceded by \cite{lsmv4}.}

\section{Introduction}

For all arities $k\geq 0$ define $\codomainbool_k$ to be the set of
functions $\{0,1\}^k\to \Rnonneg$ and define
$\codomainbool=\bigcup_k\codomainbool_k$. A function $F \in
\codomainbool_k$ is log-supermodular (lsm) if $F (\vecx \vee \vecy)F
(\vecx \wedge \vecy) \geq F(\vecx)F(\vecy)$ for all $\vecx,\vecy \in \{0,
1\}^k$ where $\vecx\vee\vecy$ and $\vecx\wedge\vecy$ denote elementwise
maximum and minimum respectively.
Denote the set of all lsm functions by $\LSM$.
Define $\IMP\in\codomainbool_2$ by
\[ \IMP(x,y)=\begin{cases}0&\text{ if $(x,y)=(1,0)$}\\
                         1&\text{ otherwise}\end{cases} \]

Bulatov et al. \cite{lsm} defined the operation of (efficient) \omegadefinability\
in order to study the computational complexity of certain
approximate counting problems.
They asked (\cite{lsm}, Section 4) whether all functions in
$\LSM$ can be defined by $\IMP$ and $\codomainbool_1$ in this sense.
We give a negative answer to this question.
{\v{Z}}ivn{\'y} et al. \cite{zivny} proved a negative result for the analogous optimisation problem.

\section{A set containing binary lsm functions}

We will construct a set $\myset$ containing binary lsm functions and closed
under a minimalist set of operations related to $T_2$-constructibility
\cite{tomo}. Later, in section 3, we will relate this to \omegadefinability.
For concision we will use vector notation such as
$\vecx=(x_1,\cdots,x_n)$,
$(\vecx,\vecy)=(x_1,\cdots,x_n,y_1,\cdots,y_m)$,
$\vecone=(1,\cdots,1)$, $(i,j,\vecx)=(i,j,x_1,\cdots,x_n)$, and
$\vecx\cdot\vecy=x_1y_1+\cdots+x_ny_n$.
For all
$k,l\geq 0$ and all $F\in\codomainbool_k$ and
$G\in\codomainbool_\ell$, define the tensor product $F\otimes
G\in\codomainbool_{k+\ell}$ by
\[ (F\otimes G)(x_1,\cdots,x_k,y_1,\cdots,y_\ell) =
    F(x_1,\cdots,x_k)G(y_1,\cdots,y_\ell) \]
For all $k\geq 2$ and all $F\in\codomainbool_k$
define the primitive contraction $\tr_{1,2} F\in\codomainbool_{k-2}$ by
\[ (\tr_{1,2}F)(x_3,\cdots,x_k) = \sum_{z=0}^1 F(z,z,x_3,\cdots,x_k) \]
For all $k\geq 0$ and all $F\in\codomainbool_k$ and all permutations
$\pi$ of $\{1,\cdots,k\}$ define the permutation
$F_\pi\in\codomainbool_k$ by
\[ F_\pi(x_1,\cdots,x_k) = F(x_{\pi(1)},\cdots,x_{\pi(k)}) \]

For all $k\geq 0$, all
$F\in\codomainbool_k$, and all $\vecw\in\{0,1\}^k$, define
$B(F,\vecw)$ by
\begin{align*}
B(F,\vecw)&=\sum_{x_1,\cdots,x_k} F(x_1,\cdots,x_k)F(1-x_1,\cdots,1-x_k)(-1)^{x_1w_1+\cdots+x_kw_k}\\
&=\sum_{\vecx} F(\vecx)F(\vecone-\vecx)(-1)^{\vecx\cdot\vecw}
\end{align*}
For all $k\geq 1$ and all $F\in\codomainbool_k$
and all $1\leq i\leq k$ and all $v\in\{0,1\}$,
define the primitive pinning $F_{i\mapsto v}$ by
\[ F_{i\mapsto v}(x_1,\cdots,x_{i-1},x_{i+1},\cdots,x_k)=F(x_1,\cdots,x_{i-1},v,x_{i+1},\cdots,x_k)\]
We say that $F'\in\codomainbool$ is a pinning of $F$ if $F'$ can be
obtained from $F$ by a (possibly empty) sequence of primitive pinnings.
Define
\begin{align*}
\myset&= \{ F\in\codomainbool : \text{ $B(F',\vecw)\geq 0$
  for all pinnings $F'$ of $F$}\\
  &\hspace{1in}\text{and all $\vecw\in\{0,1\}^k$ where $k$ is the arity of $F'$} \}
\end{align*}

\begin{lemma}\label{even_two}
Let $F$ be an arity $k$ function and let
$w_1,\cdots,w_k\in\{0,1\}$.  If $w_1+\cdots+w_k$ is odd or less than
two then $B(F,\vecw)\geq 0$.
\end{lemma}
\begin{proof}
If $w_1=\cdots=w_k=0$ then $B(F,\vecw)=\sum_{\vecx} F(\vecx)F(\vecone-\vecx)\geq 0$.
If $\vecone\cdot\vecw$ is odd then
\begin{align*}
B(F,\vecw)+B(F,\vecw)
&=\sum_{\vecx} F(\vecx)F(\vecone-\vecx)
  ((-1)^{\vecx\cdot\vecw}+(-1)^{(\vecone-\vecx)\cdot \vecw})\\
&=\sum_{\vecx} F(\vecx)F(\vecone-\vecx)
   (-1)^{\vecx\cdot\vecw}(1+(-1)^{\vecone\cdot\vecw})=0\qedhere
\end{align*}
\end{proof}

\begin{lemma}$\myset$ contains all lsm functions of arity at most 2.
\label{gotbinlsm}\end{lemma}
\begin{proof}
Let $F'$ be an lsm function of arity at most 2.
Let $F$ be a pinning of $F'$ of arity $k$
and let $\vecw\in\{0,1\}^k$.
We would like to show that $B(F,\vecw)\geq 0$.  By Lemma
\ref{even_two} we may assume $\sum_i w_i = 2$.  Hence $k=2$.
So $F=F'$ and $B(F,\vecw)=B(F,(1,1))$ is
\begin{align*}
&F(0,0)F(1,1)-F(0,1)F(1,0)-F(1,0)F(0,1)+F(1,1)F(0,0)\\
&=2(F(0,0)F(1,1)-F(0,1)F(1,0))\geq 0\qedhere
\end{align*}
\end{proof}

\begin{lemma}\label{pin_commute}
Let $k,\ell\geq 0$. Let $F\in\codomainbool_k$ and $G\in\codomainbool_\ell$.
For all pinnings $H$ of $F\otimes G$,
there exists a pinning $F'$ of $F$ and a pinning $G'$ of $G$
such that $H=F'\otimes G'$.
If $k\geq 2$ then for all pinnings $H$ of $\tr_{1,2}F$, there exists
a pinning $F'$ of $F$ such that $H=\tr_{1,2}F'$.
\end{lemma}
\begin{proof}
We argue by induction on the number of primitive pinnings
used to construct $H$.
The first statement is trivial if $H=F\otimes G$.
Otherwise, by induction $H=(F'\otimes G')_{i\mapsto v}$
for some pinnings $F'$ of $F$ and $G'$ of $G$.
If $i$ is at most the arity $k'$ of $F'$ then $H=F'_{i\mapsto v}\otimes G'$;
otherwise $H=F'\otimes G'_{(i-k')\mapsto v}$. In both cases
$H$ is of the required form.
The second statement is trivial if $H=\tr_{1,2}F$.
Otherwise, by induction $H=(\tr_{1,2}F')_{i\mapsto v}$ for some pinning $F'$ of $F$, so $H=\tr_{1,2}(F'_{(i+2)\mapsto v})$ as required.
\end{proof}

\begin{lemma}$\myset$ is closed under tensor products, primitive contractions, and permutations.
\label{fctcat}\end{lemma}
\begin{proof}
Let $F',G'\in\myset$ and let $H$ be a pinning of $F'\otimes G'$.
By Lemma \ref{pin_commute} there exist pinnings
$F$ of $F'$ and $G$ of $G'$ such that $H=F\otimes G$.
Let $k$ and $\ell$ be the arities of $F$ and $G$ respectively.
For all $a_1,\cdots,a_k,$ $b_1,\cdots,b_\ell\in\{0,1\}$,
\begin{align*}
&&&B(F\otimes G,(\veca,\vecb))\\
&=&&\sum_{\vecx,\vecy}(F\otimes G)(\vecx,\vecy)(F\otimes G)(\vecone-\vecx,\vecone-\vecy)(-1)^{\vecx\cdot\veca+\vecy\cdot\vecb}\\
&=&&\sum_{\vecx,\vecy}F(\vecx)G(\vecy)F(\vecone-\vecx)G(\vecone-\vecy)(-1)^{\vecx\cdot\veca+\vecy\cdot\vecb}\\
&=&&\sum_{\vecx}F(\vecx)F(\vecone-\vecx)(-1)^{\vecx\cdot\veca}\sum_{\vecy}G(\vecy)G(\vecone-\vecy)(-1)^{\vecy\cdot\vecb}\\
&=&&B(F,\veca)B(G,\vecb)\geq 0
\end{align*}
hence $F'\otimes G'\in\myset$.

Let $H$ be a pinning of $\tr_{1,2}F'$. By Lemma \ref{pin_commute}
there exists a pinning $F$ of $F'$ such that $H=\tr_{1,2}F$.
Let $k$ be the arity of $F$. For all
$\vecw\in\{0,1\}^{k-2}$,
\begin{align*}
&&&B(\tr_{1,2} F,\vecw)\\
&=&&\sum_{\vecx}(\tr_{1,2} F)(\vecx)(\tr_{1,2} F)(\vecone-\vecx)(-1)^{\vecx\cdot\vecw}\\
&=&&\sum_{\vecx}(F(0,0,\vecx)+F(1,1,\vecx))
(F(0,0,\vecone-\vecx)+F(1,1,\vecone-\vecx))
(-1)^{\vecx\cdot\vecw}\\
&=&&\sum_{\vecx}F(0,0,\vecx)F(0,0,\vecone-\vecx)(-1)^{\vecx\cdot\vecw}+\\
 &&&\sum_{\vecx}F(1,1,\vecx)F(1,1,\vecone-\vecx)(-1)^{\vecx\cdot\vecw}+\\
 &&&\sum_{\vecx}F(0,0,\vecx)F(1,1,\vecone-\vecx)(-1)^{\vecx\cdot\vecw}+\\
 &&&\sum_{\vecx}F(1,1,\vecx)F(0,0,\vecone-\vecx)(-1)^{\vecx\cdot\vecw}\\
&=&&B((F_{2\mapsto 0})_{1\mapsto 0},\vecw)+
    B((F_{2\mapsto 1})_{1\mapsto 1},\vecw)+\\
 &&&\sum_{i,j,\vecx}F(i,j,\vecx)F(1-i,1-j,\vecone-\vecx)
 \frac{1+(-1)^{i+j}}{2}(-1)^{\vecx\cdot\vecw}\\
&=&&B((F_{2\mapsto 0})_{1\mapsto 0},\vecw)+
    B((F_{2\mapsto 1})_{1\mapsto 1},\vecw)+\\
  &&&\hspace{0.5in}\frac{B(F,(0,0,\vecw))+B(F,(1,1,\vecw))}{2}\\
&\geq&& 0
\end{align*}
hence $\tr_{1,2}F'\in\myset$.
The definitions of $B$ and $\myset$ are clearly invariant under permutations.
\end{proof}

\begin{lemma}\label{ctform}
Let $\calF$ be a subset of $\codomainbool$ closed under
tensor products, primitive contractions, and permutations.
Let $F$ be a function of the form
\[ F(x_1,\cdots,x_n)=\sum_{x_{n+1},\cdots,x_{n+m}}\phi_1\cdots\phi_s \]
where each $\phi_j$ is a function application
$G_j(x_{i_{j,1}},\cdots,x_{i_{j,a(j)}})$
with $G_j\in\calF$, such that for $1\leq i\leq n$
exactly one pair $(j,k)$ satisfies $i_{j,k}=i$,
and for $n+1\leq i\leq n+m$
exactly two pairs $(j,k)$ satisfy $i_{j,k}=i$.
Then $F$, and hence all permutations of $F$, are in $\calF$.
\end{lemma}
\begin{proof}
We will prove the statement by induction on $m$.
Let $G=G_1\otimes \cdots \otimes G_s\in\calF$.
If $m=0$ then $F$ is a permutation of $G$ so $F\in\calF$.
Otherwise there exists a permutation $\pi$ of $\{1,\cdots,n+2m\}$
such that
\begin{align*}
&F(x_1,\cdots,x_n)\\
&=\sum_{x_{n+1},\cdots,x_{n+m}} G(x_{i_{1,1}},\cdots,x_{i_{s,a(s)}})\\
&=\sum_{x_{n+1},\cdots,x_{n+m}} G_\pi(x_{n+1},x_{n+1},x_{n+2},x_{n+2},\cdots,x_{n+m},x_{n+m},x_1,x_2,\cdots,x_n)
\end{align*}
Define $F'\in\codomainbool_{n+2}$ by
\begin{align*}
&F'(x_{n+1},x_{n+1}',x_1,\cdots,x_n)\\
&= \sum_{x_{n+2},\cdots,x_{n+m}}
G_\pi(x_{n+1},x_{n+1}',x_{n+2},x_{n+2},\cdots,x_{n+m},x_{n+m},x_1,x_2,\cdots,x_n)
\end{align*}
By the induction hypothesis $F'\in\calF$.
Hence $F=\tr_{1,2}F'\in\calF$.
\end{proof}

Define
\[ \EQ_3(x_1,x_2,x_3)=\begin{cases}1&\text{ if $x_1=x_2=x_3$}\\
                                  0&\text{ otherwise}\end{cases}\]

\begin{lemma}$\myset$ contains $\EQ_3$.\label{goteq}\end{lemma}
\begin{proof}
  Let $F$ be a pinning of $\EQ_3$. Let $k$ be the arity of $F$.
  Let $\vecw\in\{0,1\}^k$. We would like to show that
  $B(F,\vecw)\geq 0$.  By Lemma \ref{even_two} we may
  assume $\sum_i w_i = 2$.  If $F=\EQ_3$ we have
  $B(\EQ_3,\vecw)=(-1)^0+(-1)^2 \geq 0$.  Otherwise
  $F$ has arity 2 hence $F=(\EQ_3)_{i\mapsto v}$ for some $i\in\{1,2,3\}$ and $v\in\{0,1\}$.
  By symmetry we may assume $i=3$.
  But for all $x_1,x_2$ we have $\EQ_3(x_1,x_2,v)\EQ_3(1-x_1,1-x_2,v)=0$
  which implies $B((\EQ_3)_{3\mapsto v},(1,1))=0$.
\end{proof}

\section{Application to functional clones}

We will use the following definitions of \ppformulas, $\fclone{-}$,
$\fclonelim{-}$, \omegadefinability, and $F_\phi$ for
\ppformulas\ $\phi$.  These definitions are taken from \cite{lsm} (Section 2),
except that we specialise to functions in $\codomainbool$ and we rename
$\EQ$ to $\EQ_2$.

Suppose $\calF\subseteq\codomainbool$ is some collection of functions,
$V=\{v_1,\ldots,v_n\}$ is a set of variables and
$\vecx:\{v_1,\ldots,v_n\}\to\dom$ is an assignment to those variables.
An atomic formula has the form $\phi=G(v_{i_1},\ldots,v_{i_a})$ where
$G\in\calF$, $a=a(G)$ is the arity of~$G$, and
$(v_{i_1},v_{i_2},\ldots,v_{i_a})\in V^a$ is a scope.  Note that
repeated variables are allowed.  The function $F_\phi:\dom^n\to\Rnonneg$
represented by the atomic formula $\phi=G(v_{i_1},\ldots,v_{i_a})$ is
just
$$
F_\phi(\vecx)=G(\vecx(v_{i_1}),\ldots,\vecx(v_{i_a}))=G(x_{i_1},\ldots,x_{i_a}),
$$
where from now on we write $x_j=\vecx(v_j)$.

A \ppformula{} (``primitive product summation formula'') is a
summation of a product of atomic formulas.  A \ppformula~$\psi$
over~$\calF$ in variables $V'=\{v_1,\ldots,v_{n+m}\}$ has the form
\begin{equation}\label{eq:ppformulaDef}
\psi=\sum_{v_{n+1},\ldots,v_{n+m}}\,\prod_{j=1}^s\phi_j,
\end{equation}
where $\phi_j$ are all atomic formulas over~$\calF$ in the variables~$V'$.
(The variables $V$ are free, and the others, $V'\setminus V$, are bound.)
The formula~$\psi$ specifies a function $F_\psi:\dom^n\to\Rnonneg$
in the following way:
\begin{equation}\label{eq:ppfunctionDef}
F_\psi(\vecx)=\sum_{\vecy\in\dom^m}\prod_{j=1}^sF_{\phi_j}(\vecx,\vecy),
\end{equation}
where $\vecx$ and $\vecy$ are assignments
$\vecx:\{v_1,\ldots,v_n\}\to\dom$ and
 $\vecy:\{v_{n+1},\ldots,\allowbreak v_{n+m}\}\to\dom$.
The functional clone~$\fclone\calF$ generated by~$\calF$
is the set of all  functions in $\codomainbool$
that can be represented
by a \ppformula{} over~$\calF\cup\{\EQ_2\}$
where $\EQ_2$ is the binary equality function
defined by $\EQ_2(x,x)=1$ and $\EQ_2(x,y)=0$ for $x\not=y$.

Then we say that an $a$-ary function $F$ is \omegadef{} over $\calF$ if there exists a finite subset~$S_F$
of $\calF$, such that, for every $\epsilon>0$, there is an $a$-ary function $\Fhat\in\fclone{S_F}$ with
$$\|\Fhat-F\|_\infty=\max_{\vecx\in\dom^a}|\Fhat(\vecx)-F(\vecx)|<\epsilon.$$

Denote the set of functions
in $\codomainbool$ that are
\omegadef{} over~$\calF\cup\{\EQ_2\}$
by~$\fclonelim\calF$;  we call this the {\it \omegadef{} functional clone\/} generated by~$\calF$.

As in \cite{lsm} we
write $\fclone{\IMP,\codomainbool_1}$ to
mean $\fclone{\{\IMP\}\cup\codomainbool_1}$,
and write $\fclonelim{\IMP,\codomainbool_1}$ to mean
$\fclonelim{\{\IMP\}\cup\codomainbool_1}$.

\begin{lemma}\label{is_a_clone}
Let $\calF$ be a set of functions containing $EQ_3$ and closed under
tensor products, primitive contractions, and permutations.
Let $\calF'\subseteq\calF$. Then $\fclone{\calF'} \subseteq \calF$.
\end{lemma}
\begin{proof}
We need to show that each level in the definition of \ppdefinable\
function preserves membership in $\calF$: first that every atomic
formula over $\calF'\cup\{\EQ_2\}$ defines a function in $\calF$, secondly that a
product of functions in $\calF$ is in $\calF$, and finally that a
summation of functions in $\calF$ is in $\calF$.

Define $\EQ_1\in\codomainbool_1$ to be the constant function
$\EQ_1(0)=\EQ_1(1)=1$.  We have $\EQ_1=\tr_{1,2}\EQ_3\in\calF$.
Also, $\EQ_2=\tr_{1,2}(\EQ_1\otimes\EQ_3)\in\calF$.

Let $k,n\geq 0$.
Let $\phi=G(v_{i_1},\cdots,v_{i_k})$ be an atomic formula
with $G\in\calF'\cup\{\EQ_2\}$ and $i_1,\cdots,i_k\in\{1,\cdots,n\}$.
Define functions $A_0,\cdots,A_k$ by
\begin{align*}
A_0(x_1,\cdots,x_n,y_1,\cdots,y_k)&=G(y_1,y_2,\cdots,y_k)\\
A_1(x_1,\cdots,x_n,y_2,\cdots,y_k)&=G(x_{i_1},y_2,\cdots,y_k)\\
&\cdots\\
A_k(x_1,\cdots,x_n)&=G(x_{i_1},x_{i_2},\cdots,x_{i_k})
\end{align*}
for all $x_1,\cdots,x_n,y_1,\cdots,y_k\in\{0,1\}$.
We have $G\in\calF'\cup\{\EQ_2\}\subseteq\calF$, so
$A_0=\EQ_1\otimes \cdots \otimes \EQ_1 \otimes G\in\calF$.
Let $1\leq j\leq k$. Then for all $x_1,\cdots,x_n,y_{j+1},\cdots,y_k$ we have
\begin{align*}
&&&A_j(x_1,\cdots,x_n,y_{j+1},\cdots,y_k)\\
&=&&\sum_{x_{i_j}',y_j}
  A_{j-1}(x_1,\cdots,x_{i_j-1},x_{i_j}',x_{i_j+1},\cdots,x_n,y_j,\cdots,y_k)\EQ_3(y_j,x_{i_j}',x_{i_j})
\end{align*}
By Lemma \ref{ctform}, $A_{j-1},EQ_3\in\calF$ implies $A_j\in\calF$.  By induction on $j$ we have
$F_\phi=A_k\in\calF$.

Let $F,G\in \calF$ be functions of arity $n$. Then for all $\vecx\in\{0,1\}^n$,
\[(FG)(\vecx)=F(\vecx)G(\vecx)=\sum_{\vecy,\vecz} F(\vecy)G(\vecz)\EQ_3(x_1,y_1,z_1)\cdots \EQ_3(x_n,y_n,z_n)\]
By Lemma \ref{ctform}, $FG\in\calF$.

Let $F\in \calF$ be a function of arity $n+m$.
Define $F'\in\codomainbool_n$ by $F'(\vecx)=\sum_{\vecy\in\{0,1\}^m} F(\vecx,\vecy)$ for all $\vecx\in\{0,1\}^n$. Then for all $\vecx\in\{0,1\}^n$,
\[
F'(\vecx)=\sum_{\vecy} F(\vecx,\vecy)=\sum_{\vecy} F(\vecx,\vecy)\EQ_1(y_1)\cdots \EQ_1(y_m)\]
By Lemma \ref{ctform}, $F'\in\calF$.

Hence for every function $F\in\fclone{\calF'}$ we have $F\in\calF$.
\end{proof}

\begin{theorem} $\fclonelim{\IMP,\codomainbool_1}\neq \LSM$ \end{theorem}
\begin{proof}
Define $S\in\codomainbool_4$ by
\begin{align*}
S(x_1,x_2,x_3,x_4)=\begin{cases}
    4&\text{ if $x_1+x_2+x_3+x_4=4$}\\
    2&\text{ if $x_1+x_2+x_3+x_4=3$}\\
    1&\text{ otherwise}
  \end{cases}
\end{align*}

We will show that $S\in\LSM\setminus\fclonelim{\IMP,\codomainbool_1}$.
To show $S\in\LSM$, by symmetry of S and a theorem of Topkis
(\cite{lsm} Lemma 9) it suffices to show that the functions
$S(0,0,x,y)$, $S(0,1,x,y)$ and $S(1,1,x,y)$ are lsm; this is
equivalent to $1\cdot 1\geq 1\cdot 1$ and $2\cdot 1\geq 1\cdot 1$ and
$4\cdot 1\geq 2\cdot 2$ respectively.

Suppose, for contradiction, that
$S\in\fclonelim{\IMP,\codomainbool_1}$.  Then by definition of
\omegadefinability, for every $\epsilon > 0$ we can choose an arity 4
function $F_\epsilon \in \fclone{\IMP,\codomainbool_1}$ such that
\[\max_{\vecx\in\{0,1\}^4}|F_\epsilon(\vecx)-S(\vecx)|< \epsilon \]

Since $B(F,(1,1,1,1))$ is continuous in $F$ we have
\[\lim_{\epsilon\to
  0}B(F_\epsilon,(1,1,1,1))=B(S,(1,1,1,1))=4-8+6-8+4 = -2<0\]

Hence
there exists $\epsilon$ such that $B(F_\epsilon,(1,1,1,1))<0$.

Each function in $\{\IMP\}\cup\codomainbool_1$ is lsm
and of arity at most two.
By Lemmas  \ref{gotbinlsm}, \ref{fctcat}, \ref{goteq} and
\ref{is_a_clone} we have $\fclone{\IMP,\codomainbool_1}\subseteq\myset$.
Hence $F_\epsilon\in\myset$. Hence
$B(F_\epsilon,(1,1,1,1))\geq 0$, a contradiction. Therefore
$S\not\in\fclonelim{\IMP,\codomainbool_1}$.
\end{proof}

\bibliography{main}{}
\bibliographystyle{plain}

\end{document}